\documentclass[11pt]{article}
\usepackage[T1]{fontenc}
\usepackage{lmodern,amsmath,amssymb,amsthm,mathtools}
\usepackage[margin=1in]{geometry}
\usepackage{microtype,booktabs,graphicx,enumitem}
\usepackage[hidelinks]{hyperref}
\hypersetup{pdftitle={Spherical Convex Hull Membership: Complexity and Performance of Geometric and Accelerated Methods},pdfauthor={Bahman Kalantari}}
\setlist{itemsep=2pt,topsep=4pt}
\newtheorem{theorem}{Theorem}[section]
\newtheorem{lemma}[theorem]{Lemma}
\newtheorem{proposition}[theorem]{Proposition}
\newtheorem{corollary}[theorem]{Corollary}
\theoremstyle{remark}
\newcommand{\conv}{\operatorname{conv}}
\newcommand{\norm}[1]{\left\|#1\right\|}
\newcommand{\argmin}{\operatorname*{arg\,min}}

\title{Spherical Convex Hull Membership:\\
Complexity and Performance of Geometric and Accelerated Methods}
\author{Bahman Kalantari\thanks{Emeritus Professor of Computer Science,
Rutgers University.}}
\date{}
\begin{document}
\maketitle

\begin{abstract}
We study the problem of testing membership of the origin in the convex
hull of $n$ unit vectors in $\mathbb{R}^m$, seeking a convex combination of norm
at most $\varepsilon$ or a strict separating hyperplane. The Spherical Triangle
Algorithm (STA) uses nearest-point updates on pivot segments. We first
establish a conditional $O(1/\varepsilon)$ iteration bound from accumulated
geometric progress, allowing deeper pivots to compensate for shallow ones.
A three-point example shows that ordinary STA can require
$\Theta(\varepsilon^{-2})$ iterations. We show that the same conditional bound
holds for STA with Wolfe minimum-norm correction (STA--WMC), provided
the averaged geometric progress condition holds for the STA steps
taken within this corrected method. We then combine STA, WMC, and
the Nemirovski weighted-pivot method (NWP) in a hybrid with unconditional
bounds of $O(\sqrt{\log n}/\varepsilon)$ matrix--vector products with the input
matrix or its transpose and $O(mn\sqrt{\log n}/\varepsilon)$ arithmetic
operations, including corrections. We also give a self-contained
derivation of the $O(\sqrt{\log n}/\varepsilon)$ membership iteration bound
for entropic Nesterov acceleration. Computationally, we compare five
methods---ordinary STA, STA--WMC, the hybrid, standalone NWP, and
entropic Nesterov acceleration---on convex-hull queries derived from
14 public benchmark datasets and on 204 controlled synthetic inputs.
STA--WMC completes all cases in both comparisons, and the hybrid
completes all public-data cases. STA performs well on favorable
interior and separation instances; Nesterov acceleration is faster
than the STA-based methods on the highest-dimensional random
membership tests. The experiments show complementary strengths and
distinguish practical performance from worst-case guarantees.
\end{abstract}



\section{Introduction}
Let $S=\{v_1,\ldots,v_n\}\subset\mathbb{R}^m$, where
$\|v_i\|=1$ for each $i$. The spherical convex hull membership
problem asks whether
\begin{equation}\label{eq:chm}
0\in\operatorname{conv}(S).
\end{equation}
Write $V=[v_1,\ldots,v_n]$ and
$\Delta_n=\{x\in\mathbb{R}^n:x\geq0,\ \mathbf{1}^Tx=1\}$.
Throughout, $n\geq2$ and $0<\varepsilon<1$.
Unsubscripted norms are Euclidean.
An approximate answer is an explicitly represented point
$p=Vx$, $x\in\Delta_n$, with $\norm p\leq\varepsilon$.
Such an answer certifies proximity to the hull, not exact membership.
A strict separator is a vector $u$ with $v_i^Tu>0$ for every $i$.
Geometrically, STA moves along a segment toward an input vertex to
approach the origin; separation is established when all input points
lie strictly on one side of a hyperplane with the origin on the other.

As established in~\cite{Kalantari2015, KalantariZhang2022}, general convex hull membership, in both
its exact and approximate forms, reduces to the spherical problem.
We recall the normalization and recovery of convex coefficients.
If an input point equals the target $b$, membership is immediate.
Otherwise, for input points $a_i\ne b$, let
$r_i=\norm{a_i-b}$ and $v_i=(a_i-b)/r_i$.
Given $x\in\Delta_n$, define
$
 \alpha_i=(x_i/r_i)/(\sum_j x_j/r_j).
$
Then $\alpha=(\alpha_1,\ldots,\alpha_n)\in\Delta_n$ and
$
 \sum_i\alpha_i a_i-b
 ={Vx}/(\sum_j x_j/r_j).
$
Consequently, $Vx=0$ yields an exact representation of $b$,
while $\norm{Vx}\leq\varepsilon$ yields
$\left\|\sum_i\alpha_i a_i-b\right\| \leq \varepsilon\max_i\norm{a_i-b}$.

The Triangle Algorithm \cite{Kalantari2015} and its spherical form
\cite{KalantariZhang2022} use geometric pivots and nearest-point
segment updates, with distance duality supplying separation certificates.
Ordinary STA returns an approximate point or a strict separator
within $O(\varepsilon^{-2})$ iterations.
If the hull contains a radius-$\rho$ ball about zero in its linear
span, deepest pivots give $O(\rho^{-2}\log(1/\varepsilon))$ iterations.
Convex hull membership and separation also support vertex enumeration
and optimization \cite{AwasthiKalantariZhang2020,Kalantari2019Separation}.

We establish a cumulative geometric condition for $O(1/\varepsilon)$
convergence of ordinary STA: deeper pivots can compensate for shallow
ones through an average normalized-depth bound. We also consider
STA with Wolfe minimum-norm correction (WMC), denoted STA--WMC.
This correction recombines selected input points to seek a convex
combination of smaller norm. The same conditional bound holds for
STA--WMC when the averaged geometric progress condition holds for
the ordinary STA steps taken within this corrected method.
A three-point example shows that ordinary STA can nevertheless
require $\Theta(\varepsilon^{-2})$ iterations. We then combine STA, WMC,
and the Nemirovski weighted-pivot method (NWP).
A direct proof for NWP and amortized accounting of all STA and
correction work give the hybrid an unconditional
$O(mn\sqrt{\log n}/\varepsilon)$ arithmetic bound.

\paragraph{Relation to established methods.}
WMC uses Wolfe's affine minimization and deletion geometry
\cite{Wolfe1976}, with the segment initialization attributed to
Betke \cite{Betke2004} in \cite[Section~2.1]{DadushVeghZambelli2018}.
Geometry-dependent linear rates for Frank--Wolfe variants
\cite{LacosteJulienJaggi2015} can be stronger on favorable inputs,
but involve additional geometric constants.

The accelerated order $O(\sqrt{\log n}/\varepsilon)$ is available through
entropic Nesterov acceleration \cite{Nesterov2005} and the
smoothed-perceptron estimates of Ramdas and Pe\~na
\cite{RamdasPena2014}. NWP specializes Nemirovski's Mirror-Prox
construction \cite{Nemirovski2004}. Alongside the cumulative STA
analysis and sharp slow example, we give self-contained proofs of the
specialized Nesterov membership bound and the NWP approximate-point-or-separator
bound, and prove the hybrid's accelerated guarantee with complete
correction-work accounting. These derivations make the certificates,
assumptions, and implementation costs explicit; the underlying
first-order methods and their accelerated order are established.
Table~\ref{tab:guarantees} summarizes the guarantees.
Earlier computational comparisons appear in \cite{Filippozzi2023}.

\paragraph{Computational scope.}
We compare five methods: ordinary STA, NWP, entropic Nesterov
acceleration, the hybrid, and STA--WMC with remembered pivots. The experiments examine when STA steps suffice,
when correction pays for its cost, and how the safeguard affects
runtime and completion within fixed budgets.

\paragraph{Work model and organization.}
We count full products by $V$ or $V^T$, each costing $O(mn)$,
and all additional arithmetic, including amortized affine work.
The bounds use exact real arithmetic with unit-cost scalar
exponentials and logarithms; the experiments use double precision.
Sections~\ref{sec:sta}--\ref{sec:amortized} analyze STA;
Sections~\ref{sec:backup} and~\ref{sec:comparison} describe NWP and
Nesterov acceleration. Sections~\ref{sec:hybrid}--\ref{sec:modified}
develop the combinations, Section~\ref{sec:computations} reports the
experiments, and Section~\ref{sec:practice} gives conclusions and
practical guidance. Appendix~\ref{sec:synthetic-details} contains
supporting Nesterov proofs.

\section{STA: geometry, progress, and a sharp slow trajectory}
\label{sec:sta}
At a nonzero point $p=Vx$, $x\in\Delta_n$, an ordinary
\emph{pivot} is an input vertex $v\in S$ satisfying
$\norm{p-v}\geq\norm{v}$, or equivalently,
$p^Tv\leq\norm{p}^2/2$.
Following the terminology of the original Triangle Algorithm
in~\cite{Kalantari2015}, a point $p\in\conv(S)$ admitting no such pivot is called
a \emph{witness}. In this case,
$p^Tv_i>\norm{p}^2/2$ for every $i$,
so the perpendicular bisector hyperplane of the segment joining
the origin to $p$,
$H=\left\{z\in\mathbb{R}^m: p^Tz=\norm{p}^2/2\right\}$,
strictly separates the origin from $\conv(S)$.
Moreover, writing $d_*=\operatorname{dist}(0,\conv(S))$, a witness gives
\[
 \tfrac12\norm p\leq d_*\leq\norm p.
\]
Indeed, every $z\in\conv(S)$ satisfies
$\norm p\,\norm z\geq p^Tz>\norm p^2/2$, and $p\in\conv(S)$.

A \emph{strict pivot}, as used in spherical TA, satisfies
the stronger condition $p^Tv\leq0$; the weak inequality is
part of this convention. Every strict pivot is an ordinary pivot,
since $\norm{p-v}^2=\norm p^2+1-2p^Tv\geq1$.
A deepest pivot minimizes the scores $p^Tv_i$.
If their minimum $c=\min_i p^Tv_i$ is positive, $p$ gives a
strict separating direction: the hyperplane $p^Tz=c/2$
separates the origin from $\conv(S)$, and
\[
 \frac{c}{\norm p}\leq d_*\leq\norm p.
\]
However, $c>0$ need not imply $c>\norm p^2/2$.
Thus failure to find a strict pivot gives a separator, but the
perpendicular-bisector conclusion and the $\norm p/2$ lower bound
require the stronger witness test above.
The implementations studied here stop at either an approximate
point or a positive minimum score.

If $0\in\conv(S)$, a strict pivot always exists:
for any $x^*\in\Delta_n$ with $Vx^*=0$,
$\sum_i x_i^*p^Tv_i=0$,
so the scores cannot all be positive.

With a strict pivot $v$, STA replaces $p$ by the point of $[p,v]$
nearest the origin:
\begin{equation}\label{eq:sta}
 p^+=(1-t)p+tv,\qquad
 t=\frac{\norm p^2-p^Tv}{\norm{v-p}^2}.
\end{equation}
Here $0<t<1$, because the numerator is positive and the denominator
exceeds it by $1-p^Tv>0$. Original convex coefficients
are updated by $x^+=(1-t)x+te_j$ if $v=v_j$.
The segment operation costs $O(m+n)$ beyond the score scan.
No projection onto the full hull is required.

\begin{lemma}[Exact progress]\label{lem:progress}
Let $\delta=\norm p>0$ and $s=-p^Tv/\delta\geq0$. Then
\begin{equation} \label{eq:residual}
 t=\frac{\delta(\delta+s)}{1+\delta^2+2\delta s}, \quad 
 \delta^2-\norm{p^+}^2
 =\frac{\delta^2(\delta+s)^2}{1+\delta^2+2\delta s}, \quad
 \norm{p^+}^2=\frac{\delta^2(1-s^2)}{1+\delta^2+2\delta s}.
\end{equation}
\end{lemma}
\begin{proof}
Substitute $p^Tv=-\delta s$ and $\norm v=1$ into
\eqref{eq:sta}. Orthogonal projection onto the line through $p,v$
reduces squared norm by
$(\norm p^2-p^Tv)^2/\norm{v-p}^2$. Subtracting this expression from $\delta^2$
gives the last equation in \eqref{eq:residual}.
\end{proof}

\begin{proposition}[Approximation or separation]\label{prop:quadratic}
Write $\delta_k=\norm{p_k}$ for the residual after $k$ strict-pivot
STA updates. For any input,
\[
 \delta_k\leq(\delta_0^{-2}+k)^{-1/2},
\]
unless the algorithm has already terminated. Consequently, STA
returns a convex combination of norm at most $\varepsilon$ or a strict
separator within $O(\varepsilon^{-2})$ iterations and $O(mn/\varepsilon^2)$
arithmetic operations.
\end{proposition}
\begin{proof}
Equation~\eqref{eq:residual} implies
$\delta_{k+1}^2\leq\delta_k^2/(1+\delta_k^2)$ whenever a strict
pivot is used. Taking reciprocals and summing proves the bound.
At each nonterminal point, either all scores are positive and give
a separator, or a strict pivot is available. After at most
$\lceil\varepsilon^{-2}\rceil$ such updates the residual is at most $\varepsilon$.
Each update and its full score scan cost $O(mn)$.
\end{proof}

\subsection{A three-point example}
Consider $S=\{(-1,0),(1,0),(0,1)\}$ and start at $p_0=(0,1)$.
Either horizontal vertex can be the first pivot. Thereafter deepest
pivots alternate between the two horizontal vertices. Put
$\eta_k=\norm{p_k}^2$. For $k\geq1$ the iterates have the form
$p_k=(\sigma_k\eta_k,\sqrt{\eta_k-\eta_k^2})$, where $\eta_1=\tfrac12$ and $\sigma_k\in\{-1,1\}$.
For the opposite horizontal pivot, direct substitution in
\eqref{eq:sta} gives
\begin{equation}\label{eq:zigzag}
 \eta_{k+1}=\frac{\eta_k(1-\eta_k)}{1+3\eta_k},\qquad
 \frac1{\eta_{k+1}}-\frac1{\eta_k}=\frac4{1-\eta_k}.
\end{equation}
The recurrence preserves the displayed form and reverses $\sigma_k$.
Since $0<\eta_k\leq1/2$, the reciprocal increments lie between $4$ and
$8$. Thus $\eta_k=\Theta(1/k)$ and the residual accuracy complexity is
$\Theta(\varepsilon^{-2})$ even with deepest pivots.

The two cycling pivots have midpoint zero. Remembering and combining
them resolves the example immediately, whereas the ordinary segment
update repeatedly loses that opportunity. This motivates geometric
memory, but resolving this example alone does not establish a faster
general bound.

\subsection{Interior geometry}
If the hull contains a radius-$\rho$ ball about zero in its linear
span, a deepest pivot $v_{i_k}$ at $p_k$ has opposite depth
$s_k=-p_k^Tv_{i_k}/\delta_k\geq\rho$. Since
$0<\delta_k\leq1$ and $0\leq s_k\leq1$, the denominator in
\eqref{eq:residual} is at most $4$. Consequently
$\delta_{k+1}^2\leq(1-\rho^2/4)\delta_k^2$, giving
$O(\rho^{-2}\log(1/\varepsilon))$ iterations. This familiar favorable
regime \cite{Kalantari2015, KalantariZhang2022} does not cover boundary membership
such as \eqref{eq:zigzag}. The next section gives an intermediate,
trajectory-based regime.

\section{Conditional amortized
\texorpdfstring{$O(1/\varepsilon)$}{O(1/epsilon)}
convergence of ordinary STA}
\label{sec:amortized}
Let $v_{i_k}$ be the strict pivot chosen at iteration $k$, and set
$s_k=-p_k^Tv_{i_k}/\delta_k$. The quantity relevant to reciprocal residual progress is
\begin{equation}\label{eq:depthscore}
 a_k=\frac{s_k^2}{\delta_k}
     =\frac{(p_k^Tv_{i_k})^2}{\norm{p_k}^3}.
\end{equation}
It is available from the pivot score and the current residual.

\begin{lemma}[Reciprocal progress]\label{lem:reciprocal}
For an STA update whose resulting residual is positive,
\begin{equation}\label{eq:reciprocal}
 \delta_{k+1}^{-1}-\delta_k^{-1}
 \geq s_k+\frac{s_k^2}{2\delta_k}\geq\frac12 a_k.
\end{equation}
Hence every finite nonzero trajectory satisfies
\begin{equation}\label{eq:finite-amortized}
 \delta_N\leq\left(\delta_0^{-1}+
                   \frac12\sum_{k=0}^{N-1}a_k\right)^{-1}.
\end{equation}
\end{lemma}
\begin{proof}
Write $\delta=\delta_k$ and $s=s_k$. A positive next residual implies
$0\leq s<1$. Since
$1+\delta^2+2\delta s-(1+\delta s)^2=\delta^2(1-s^2)\geq0$
and $(1-s^2)^{-1/2}\geq1+s^2/2$, the exact update gives
\[
 \delta_{k+1}^{-1}
 =\frac{\sqrt{1+\delta^2+2\delta s}}{\delta\sqrt{1-s^2}}
 \geq\delta^{-1}(1+\delta s)(1+s^2/2).
\]
Subtract $\delta^{-1}$ and discard the nonnegative term $s^3/2$.
Summing proves \eqref{eq:finite-amortized}. If $s=1$, the update reaches
zero and no reciprocal expression is needed.
\end{proof}

\begin{theorem}[Amortized opposite-depth bound]\label{thm:amortized}
Suppose an STA membership trajectory satisfies, for constants
$\gamma>0$ and integer $N_0\geq1$ independent of the requested
tolerance $\varepsilon$,
\begin{equation}\label{eq:amortized-hyp}
 \frac1N\sum_{k=0}^{N-1}a_k\geq\gamma,
 \qquad N\geq N_0,
\end{equation}
as long as the trajectory has not reached zero. Then
\[
 \delta_N\leq\frac1{\delta_0^{-1}+\gamma N/2},\qquad N\geq N_0.
\]
If $T(\varepsilon)$ is the first index with $\delta_k\leq\varepsilon$, then
\begin{equation}\label{eq:amortized-count}
 T(\varepsilon)\leq\max\left\{N_0,
 \left\lceil\frac2\gamma
 (\varepsilon^{-1}-\delta_0^{-1})_+\right\rceil\right\},
 \qquad (r)_+=\max\{r,0\}.
\end{equation}
In particular $T(\varepsilon)=O(N_0+1/(\gamma\varepsilon))$.
\end{theorem}
\begin{proof}
Insert \eqref{eq:amortized-hyp} into
\eqref{eq:finite-amortized} and solve the resulting residual bound
for $N$. Earlier exact termination only improves the stopping bound.
\end{proof}

This is an analysis of ordinary STA, with no change to its update.
The condition is weaker than requiring $a_k\geq\gamma$ at every step:
poor steps may be compensated by other steps. A finite observed average
certifies the finite estimate \eqref{eq:finite-amortized}; it does not
guarantee the same average in the future. A uniform complexity statement
over an input class requires uniform control of both $N_0$ and $\gamma$.
For a fixed trajectory these constants may depend on the input and
pivot choices, but not on $\varepsilon$. The condition is sufficient, not
claimed necessary. Identifying boundary input classes with uniform
constants remains open.

For example, $s_k\geq c\sqrt{\delta_k}$ eventually implies
$O(1/\varepsilon)$ accuracy dependence. More generally, if every consecutive
block of $B$ steps after a fixed transient has
$\sum a_k\geq\gamma B$, summing \eqref{eq:reciprocal} over complete
blocks gives $O(B+1/(\gamma\varepsilon))$ additional steps. The remaining
partial block changes only the additive term. On the three-point
trajectory, however, $s_k=\delta_k$ and $a_k=\delta_k\to0$;
the amortized condition fails, consistently with its sharp slow rate.

The statistic \eqref{eq:depthscore} motivates monitoring reciprocal
progress. A safeguard is still needed when cumulative depth is
insufficient; Section~\ref{sec:hybrid} supplies one without a
trajectory-quality assumption.

\subsection{Relation to earlier pointwise conditions}
The earlier STA analysis \cite{KalantariZhang2022} gives
$O(M/\varepsilon)$ iterations when each non-witness iterate above the
requested tolerance admits, and selects, a pivot satisfying
$\norm{p_k-v_{i_k}}\geq\sqrt{1+\varepsilon/M}$, with $M\geq1$.
Equivalently,
\begin{equation}\label{eq:prior-epsilon-property}
 \delta_k^2+2\delta_k s_k\geq\varepsilon/M.
\end{equation}
For $M$ independent of $\varepsilon$, this already gives conditional
$O(1/\varepsilon)$ complexity.

Condition~\eqref{eq:amortized-hyp} instead controls the average
of $a_k=s_k^2/\delta_k$ along the trajectory, permitting compensation
between iterations. It relaxes a pointwise lower bound on $a_k$,
but is not asserted to be weaker than
\eqref{eq:prior-epsilon-property}, which depends on the tolerance.
Indeed, the earlier condition is automatic for strict pivots when
$\delta_k^2\geq\varepsilon/M$. Theorem~\ref{thm:amortized} supplies an
explicit cumulative progress certificate and its iteration bound.

\section{The Nemirovski weighted-pivot method}
\label{sec:backup}

The \emph{Nemirovski weighted-pivot method} (NWP) specializes
Nemirovski's Mirror-Prox framework \cite{Nemirovski2004} to spherical
convex hull membership. It can run independently or serve as the
hybrid's safeguard, returning an explicit approximate convex
combination or a strict separator.

\subsection{Algorithm}

For $z\in\mathbb{R}^m$, define
\begin{equation}\label{eq:weights}
 F(z)=\log\sum_{i=1}^n e^{-v_i^Tz},\qquad
 w_i(z)=\frac{e^{-v_i^Tz}}{\sum_{j=1}^n e^{-v_j^Tz}},
 \qquad b(z)=Vw(z).
\end{equation}
The point $b(z)\in\conv(S)$ assigns greater weight to vertices
having smaller inner products with $z$. We call $b(z)$ a
\emph{weighted pivot}; it need not be a strict STA pivot.

Initialize
$z_0=r_0=0\in\mathbb{R}^m$ and $c_0=0\in\mathbb{R}^n$.
For $k=0,1,\ldots$, perform the following steps.
\begin{enumerate}
\item Compute the two weighted pivots
\begin{equation}\label{eq:two-pivots}
 b_k=b(z_k),\qquad q_k=b(z_k+r_k).
\end{equation}

\item Update the state using the values from iteration $k$:
\begin{equation}\label{eq:weighted-update}
 \begin{aligned}
 z_{k+1}&=z_k+r_k+b_k,\\
 r_{k+1}&=r_k+q_k,\\
 c_{k+1}&=c_k+w(z_k+r_k).
 \end{aligned}
\end{equation}

\item If $\norm{r_{k+1}}/(k+1)\leq\varepsilon$, return the convex
combination with coefficient vector $c_{k+1}/(k+1)$.
Otherwise, if
$\min_{1\leq i\leq n}v_i^Tz_{k+1}>0$,
return the strict separator $z_{k+1}$.
\end{enumerate}
For a separating direction $z$, put $c=\min_i z^Tv_i>0$;
the hyperplane $z^Ty=c/2$ strictly separates the origin from the hull.

Because $r_0=Vc_0=0$, the updates preserve $r_k=Vc_k$.
Moreover, $c_k$ is the sum of $k$ probability vectors. Thus,
for every $k\geq1$,
\begin{equation}\label{eq:average}
 \bar q_k:=\frac{r_k}{k}
 =V\frac{c_k}{k}
 =\frac1k\sum_{j=0}^{k-1}q_j,
 \qquad \frac{c_k}{k}\in\Delta_n.
\end{equation}
The candidate supplied to the hybrid is this average $\bar q_k$,
together with its coefficients, rather than an individual
weighted pivot. The hybrid retains the NWP state between
activations and applies its own acceptance test.

\subsection{Convergence and complexity}

To establish the convergence and iteration bound of NWP, we consider the quantity
$\Phi_k:=F(z_k)+\tfrac12\norm{r_k}^2$.
We show that $\Phi_k$ never increases. This bounds the norm
of the accumulated vector $r_k$ unless $z_k$ already provides
a strict separator.

\begin{lemma}\label{lem:smooth}
The function $F$ in \eqref{eq:weights} is convex,
$\nabla F(z)=-b(z)$, and its gradient is $1$-Lipschitz
in Euclidean norm. For all $x,u,h\in\mathbb{R}^m$,
\begin{align}
 F(x+h)&\leq F(x)+\nabla F(x)^Th
             +\tfrac12\norm h^2,
             \label{eq:smooth-upper}\\
 F(x)&\geq F(u)+\nabla F(u)^T(x-u)
             +\tfrac12\norm{\nabla F(x)-\nabla F(u)}^2.
             \label{eq:smooth-lower}
\end{align}
\end{lemma}

\begin{proof}
Differentiation gives $\nabla F(z)=-b(z)$ and
\[
 0\leq h^T\nabla^2F(z)h
 =\sum_i w_i(z)(h^Tv_i)^2
  -\left(\sum_i w_i(z)h^Tv_i\right)^2
 \leq\norm h^2.
\]
The lower bound is nonnegativity of the weighted variance;
the upper bound uses $\norm{v_i}=1$.
Consequently, $F$ is convex with a $1$-Lipschitz gradient,
which gives \eqref{eq:smooth-upper}.

For \eqref{eq:smooth-lower}, fix $u$ and set
$G(t)=F(t)-\nabla F(u)^Tt$.
Then $G$ is convex with a $1$-Lipschitz gradient, and
$\nabla G(u)=0$, so $u$ minimizes $G$.
Applying \eqref{eq:smooth-upper} to $G$ gives
\[
 G(u)\leq G(x-\nabla G(x))
       \leq G(x)-\tfrac12\norm{\nabla G(x)}^2.
\]
Since $\nabla G(x)=\nabla F(x)-\nabla F(u)$, the preceding
inequality gives
\[
 \begin{aligned}
 F(x)-F(u)-\nabla F(u)^T(x-u)
 &=G(x)-G(u)\\
 &\geq\tfrac12\norm{\nabla F(x)-\nabla F(u)}^2.
 \end{aligned}
\]
Adding $F(u)+\nabla F(u)^T(x-u)$ to both sides yields
\eqref{eq:smooth-lower}.
\end{proof}

\begin{theorem}[NWP convergence and complexity]
\label{thm:weighted}
The NWP iterates satisfy, for every $k\geq0$,
\begin{equation}\label{eq:potential}
 F(z_k)+\tfrac12\norm{r_k}^2\leq\log n.
\end{equation}
If $0\in\conv(S)$, then, for every $k\geq1$,
\begin{equation}\label{eq:weighted-rate}
 \norm{\bar q_k}\leq\frac{\sqrt{2\log n}}{k}.
\end{equation}
For any input and any $\varepsilon>0$, NWP returns a convex
combination of norm at most $\varepsilon$ or a strict separator
within
\begin{equation}\label{eq:L}
 L=\left\lceil\frac{\sqrt{2\log n}}{\varepsilon}\right\rceil
\end{equation}
updates. Each update requires $O(mn)$ arithmetic operations.
Auxiliary storage is $O(m+n)$, excluding the input matrix.
\end{theorem}

\begin{proof}
For each $k\geq0$, applying \eqref{eq:smooth-lower} with
$x=z_k$ and $u=z_k+r_k$, and using
$\nabla F(z_k)=-b_k$ and $\nabla F(z_k+r_k)=-q_k$, gives
\[
 F(z_k+r_k)\leq
 F(z_k)-q_k^Tr_k-\tfrac12\norm{q_k-b_k}^2.
\]
Applying \eqref{eq:smooth-upper} at $z_k+r_k$ with increment
$b_k$ yields
\[
 \begin{aligned}
 F(z_{k+1})
 &=F(z_k+r_k+b_k)\\
 &\leq F(z_k+r_k)-q_k^Tb_k+\tfrac12\norm{b_k}^2\\
 &\leq F(z_k)-q_k^Tr_k-\tfrac12\norm{q_k}^2.
 \end{aligned}
\]
Since $r_{k+1}=r_k+q_k$, we obtain
\[
 F(z_{k+1})+\tfrac12\norm{r_{k+1}}^2
 \leq F(z_k)+\tfrac12\norm{r_k}^2.
\]
The initial value is $F(z_0)+\tfrac12\norm{r_0}^2=\log n$,
proving \eqref{eq:potential}.

If $0\in\conv(S)$, then for every $z$ at least one
$v_i^Tz$ is nonpositive. Hence $F(z)\geq0$, and
\eqref{eq:potential} implies
$\norm{r_k}\leq\sqrt{2\log n}$.
Dividing by $k$ proves \eqref{eq:weighted-rate}.

For an arbitrary input, suppose the method has not stopped
before update $L$. If $\norm{r_L}/L\leq\varepsilon$, it returns
the approximate convex combination.
Otherwise,
$\norm{r_L}>L\varepsilon\geq\sqrt{2\log n}$,
so \eqref{eq:potential} implies $F(z_L)<0$.
Therefore
$\sum_{i=1}^n e^{-v_i^Tz_L}<1$,
which forces $v_i^Tz_L>0$ for every $i$.
Thus $z_L$ passes the strict-separation test.
\end{proof}

\paragraph{Implementation and work per update.}
Retain $\zeta_k=V^Tz_k$ and $b_k$ from one iteration to the next.
Two products by $V^T$ give
$g_k=V^Tr_k$ and $h_k=V^Tb_k$. The scores $\zeta_k+g_k$ give
$w(z_k+r_k)$, and one product by $V$ gives $q_k$.
The updated scores $\zeta_{k+1}=\zeta_k+g_k+h_k$ supply the separator
test and $w(z_{k+1})$; a second product by $V$ gives $b_{k+1}$.
Thus an update uses four full products and $O(m+n)$ additional
work, including coefficient accumulation, with $O(m+n)$ auxiliary
storage. Subtract the largest exponent before evaluating normalized
exponentials.

\section[Entropic Nesterov acceleration for spherical convex hull membership]{Entropic Nesterov acceleration for spherical\\convex hull membership}
\label{sec:comparison}
The Nesterov method used here specializes \cite{Nesterov2005} to
$f(x)=\tfrac12\norm{Vx}^2$ on $\Delta_n$,
using the $\ell_1$ norm and entropy
$d(x)=\log n+\sum_i x_i\log x_i$, with $0\log0=0$.
This section gives the implemented updates and their membership
bound; the supporting derivations are in
Appendices~\ref{app:nesterov-estimate} and~\ref{app:hard}.

\subsection{Algorithm and membership bound}
Initialize $x_0=\mathbf{1}/n$ and $G_{-1}=0$. For $k\geq0$, compute
\begin{align}
 g_k&=V^TVx_k,\nonumber\\
 y_k&\in\argmin_{y\in\Delta_n}
 \{g_k^T(y-x_k)+\tfrac12\norm{y-x_k}_1^2\},
 \label{eq:nesterov-hard}\\
 G_k&=G_{k-1}+\frac{k+1}{2}g_k,\qquad
 (z_k)_i=\frac{e^{-(G_k)_i}}{\sum_j e^{-(G_k)_j}},
 \label{eq:nesterov-soft}\\
 x_{k+1}&=\frac{k+1}{k+3}y_k+\frac2{k+3}z_k.
 \label{eq:nesterov-coupling}
\end{align}
Here $x_k,y_k,z_k$ are Nesterov simplex iterates, distinct from
the Nemirovski iterates. The update for $y_k$ transfers weight from large-score
coordinates to a minimum-score coordinate. Sorting the donors
and minimizing a piecewise quadratic gives the exact update;
Appendix~\ref{app:hard} supplies the procedure and proof.

In the implementation, the nearest point to the origin on
$[Vy_k,Vz_k]$ is also offered as an output, with its explicit
convex coefficients. Return an output when its norm is at most
$\varepsilon$, or a candidate direction $u$ when a full score test verifies
$\min_i v_i^Tu>0$. The segment refinement changes only the output,
not the recurrence~\eqref{eq:nesterov-coupling}.

The unit-column assumption implies
\begin{equation}\label{eq:l1-smoothness}
 \norm{\nabla f(x)-\nabla f(y)}_\infty
 \leq\norm{V(x-y)}_2\leq\norm{x-y}_1.
\end{equation}
Entropy has minimum zero at $\mathbf{1}/n$, is bounded above by
$\log n$, and is $1$-strongly convex in $\ell_1$:
in the interior, $h^T\nabla^2d(x)h=\sum_i h_i^2/x_i
\geq(\sum_i|h_i|)^2$; integration and continuity extend this to
the boundary.
For the iterates \eqref{eq:nesterov-hard}--\eqref{eq:nesterov-coupling}
and any minimizer $x^*$ of $f$ on $\Delta_n$, the convergence
estimate is
\[
 f(y_k)-f(x^*)\leq
 \frac{4d(x^*)}{(k+1)(k+2)}
 \leq\frac{4\log n}{(k+1)(k+2)}.
\]
A self-contained proof is given in Appendix~\ref{app:nesterov-estimate}.
On membership inputs $f(x^*)=0$, so this estimate gives
\begin{equation}\label{eq:nesterov-bound}
 \norm{Vy_k}\leq\sqrt{\frac{8\log n}{(k+1)(k+2)}}.
\end{equation}
Thus $O(\sqrt{\log n}/\varepsilon)$ iterations suffice, and the output
refinement preserves this bound.

\subsection{Separation and computational cost}
In the computational experiments, Nesterov's method also tests
candidate directions for strict separation. The membership estimate
above does not establish an unconditional accelerated iteration bound
for these separation checks.
The same $O(\sqrt{\log n}/\varepsilon)$ membership iteration bound can
also be deduced from the smoothed-perceptron estimates of
Ramdas and Pe\~na \cite[Lemmas~4--5]{RamdasPena2014}.

Including output refinement, the entropic implementation uses
three full input products and $O(n\log n)$ sorting work per
iteration, for $O(mn+n\log n)$ arithmetic and $O(m+n)$ auxiliary
storage. NWP uses four full products and $O(mn)$ arithmetic.
For $m\geq\log n$ the per-iteration costs have the same asymptotic
order.

\section{Safeguarding STA with WMC and NWP}
\label{sec:hybrid}
Ordinary STA remains the main procedure. After a slow block of $B$
steps, the hybrid first attempts WMC and activates the NWP safeguard
only if progress remains insufficient. Whereas STA moves toward one
vertex at a time, WMC retains and recombines a small active collection.
In the three-point example, combining the two alternating pivots
finds the origin immediately. WMC uses Wolfe's affine minimization
and deletion procedure \cite{Wolfe1976}, with the segment
initialization attributed to Betke \cite{Betke2004} in
\cite[Section~2.1]{DadushVeghZambelli2018}.

\subsection{WMC with retained state}
At the beginning of each WMC scan, maintain affinely independent
input points $a_j=v_{i_j}$, $j=1,\ldots,r$, whose affine
minimum-norm point belongs to their convex hull. Writing
$A=[a_1,\ldots,a_r]$, this point therefore has the representation
$a=A\lambda$ with $\lambda\in\Delta_r$.
This property holds initially by taking $A$ to consist of the first
STA pivot alone and $\lambda=(1)$. After a new point is inserted,
the affine minimum-norm coefficients may have negative entries;
the correction and deletion steps below restore the stated property.
After removing zero-weight points, all coefficients are positive;
following Wolfe, the resulting active set is a \emph{corral}.
Since $a$ already minimizes the norm over its affine hull,
further progress requires a new input point.
Retain the corral between scans; intervening STA pivots are not
automatically collected.

The following operation is one WMC scan.
\begin{enumerate}
\item At the corral point $a$, return its convex coefficients if
      $\norm a\leq\varepsilon$. Otherwise, select $v_j$ minimizing $a^Tv_i$
      over all input points. If $a^Tv_j>0$, return $a$ as a strict
      separating direction. Either return terminates the whole
      algorithm. If neither test succeeds, continue to Step~2.
\item Append $v_j$ and initialize feasible coefficients
at the point on the segment $[a,v_j]$ nearest to the origin.
\item Compute the affine minimum-norm coefficients $\alpha$ of the
      enlarged active set. If they are nonnegative, accept them and
      remove all zero-weight points.
      Otherwise move the feasible coefficients toward $\alpha$ until
      the first negative-target coordinate reaches zero; delete that
      coordinate and repeat this step on the reduced set.
\end{enumerate}
If neither stopping test in Step~1 succeeds, then
$a^Tv_j\leq0<\norm a^2$, so the selected vertex strictly improves
the corral point. No further optimality test is needed.
The stopping test for separation uses all input points.

\begin{samepage}
\begin{lemma}[Validity and amortized WMC cost]\label{lem:corral}
WMC preserves feasibility and does not increase the norm.
Every appended input point is affinely independent of the current
corral. Over $J$ corrective scans, there are at most $J$ insertions,
$J$ deletions, and $2J$ affine minimizations. Using the updated
matrix factorization described below, the total arithmetic cost
is $O(Jmn)$.
\end{lemma}
\end{samepage}

\begin{proof}
At the affine minimum-norm point $a$, every point $a_i$ of the
corral satisfies $a^Ta_i=\norm{a}^2$. Thus every point in its
affine hull has the same inner product with $a$. An input point
$v_j$ satisfying $a^Tv_j<\norm{a}^2$ lies outside this affine hull.
The enlarged active set is therefore affinely independent and
has size at most $\min(m+1,n)$.

The initialization at the point of $[a,v_j]$ nearest to the
origin is feasible and reduces the norm. The affine minimizer
of any current active set minimizes squared norm over its
entire affine hull. If its coefficients are nonnegative,
the correction is complete. Otherwise, moving the current
coefficients toward the affine-minimizer coefficients does not
increase squared norm until a coefficient first reaches zero.
The move can have zero length when a current zero coefficient has
a negative target. Simultaneous zero coefficients may be deleted
together or in successive zero-length moves.
Deleting that point preserves feasibility and the represented
point. Each unsuccessful affine minimization therefore removes
a point, so this process terminates.

Starting with one point and retaining at least one, the total
number of deletions cannot exceed the number of insertions.
Each insertion requires an affine minimization, and each
deletion requires at most one further affine minimization.
This proves the stated counts.

To perform these minimizations efficiently, maintain a QR
factorization of the matrix of augmented columns:
\[
M=
\begin{bmatrix}
a_1 & \cdots & a_r\\
1   & \cdots & 1
\end{bmatrix}
=QR.
\]
Here $Q$ has $r$ orthonormal columns and $R$ is an invertible
$r\times r$ upper-triangular matrix. Compute
\[
 R^Tt=\mathbf{1},\qquad Rh=t,\qquad
 \alpha=\frac{h}{t^Tt}.
\]
Indeed, $M^TM=A^TA+\mathbf{1}\mathbf{1}^T$ and
$\mathbf{1}^Th=t^Tt>0$. On $\mathbf{1}^T\alpha=1$, minimizing
$\norm{M\alpha}^2$ is equivalent to minimizing $\norm{A\alpha}^2$,
so the displayed vector is the affine minimizer.
Update the factorization whenever a point is inserted
or deleted, rather than recomputing it from the beginning.
Each update costs $O(mr)$ operations, and the triangular
solves cost $O(r^2)$.

Since $r\leq\min(m+1,n)$, both costs are $O(mn)$.
Reconstructing coefficients with respect to the original input
points costs at most $O(nr)=O(mn)$, and each full corrective
score scan also costs $O(mn)$. Summing over the scans,
insertions, deletions, and affine minimizations gives
the total cost $O(Jmn)$.
\end{proof}

\subsection{Combining STA, WMC, and the NWP safeguard}
Fix a block size $B\geq1$, independently of $\varepsilon,m,n$.
The experiments use $B=8$.
Start from any $x\in\Delta_n$ and $p=Vx$.
Keep the WMC corral and the NWP sequence as separate persistent states.
Repeat the following cycle, stopping immediately at an approximate
answer or a strict separator.
\begin{enumerate}
\item Record $d_0=\norm p$ and take $B$ ordinary deepest-pivot STA
      steps, unless a stopping test succeeds earlier.
\item If $1/\norm p-1/d_0<B$, allow up to $B$ WMC
      scans as above. Accept a corral point whenever its norm is smaller
      than $\norm p$. End this correction phase early if the current
      residual is at most half its value before the phase.
\item If the total reciprocal increase is still less than $B$,
      take $B$ NWP updates \eqref{eq:weighted-update}, stopping
      earlier on accuracy or separation. Initialize $z=r=c=0$ only
      on the first activation. Check the cumulative average after
      every NWP update.
\item After each NWP update in the preceding step, the cumulative
      average replaces the main point if it halves its norm.
      Always accept an average meeting $\varepsilon$, even if the halving
      criterion fails. Resume ordinary STA after the NWP phase.
\end{enumerate}
The original coefficients accompany every accepted point. The corral
is not reinitialized from the accepted point, and the NWP state
is neither restarted nor overwritten. These provisions preserve both
the WMC cost accounting and the potential proof.

\begin{theorem}[Complexity of the hybrid algorithm for
approximate membership or separation]\label{thm:hybrid}
The specified hybrid returns an approximate convex combination or a
strict separator. Let $K$ be the number of NWP updates and $N$
the number of ordinary STA steps. With $L$ from \eqref{eq:L},
\begin{equation}\label{eq:hybrid-count}
 K\leq L,\qquad N<\frac1{\varepsilon}+K+2B.
\end{equation}
There are at most $N$ corrective scans and at most $2N$ affine solves.
For fixed $B$ the total number of full input products is
$O(\sqrt{\log n}/\varepsilon)$ and the total arithmetic work is
\begin{equation}\label{eq:hybrid-work}
 O\!\left(mn\frac{\sqrt{\log n}}\varepsilon\right).
\end{equation}
Working storage, including the input and WMC, is
$O(mn)$. These bounds apply whether the origin lies inside, on
the boundary of, or outside the convex hull.
\end{theorem}
\begin{proof}
The main residual never increases. Call a completed nonterminal block
\emph{good} if its STA steps and WMC correction phase increase the reciprocal
residual by at least $B$. Before termination $1/\norm p<1/\varepsilon$.
Summing the increases over all good blocks shows that they account for
fewer than $1/\varepsilon$ ordinary steps.

Every other complete nonterminal block advances the NWP sequence
by $B$ updates. Its total number of updates is at most $L$, since
Theorem~\ref{thm:weighted} applies to that persistent sequence
regardless of intervening STA steps: its potential depends only on
its own unchanged state and update count. Charging these blocks to NWP
updates and allowing a terminal block and a final incomplete block
gives \eqref{eq:hybrid-count}. The same argument rules out an infinite
run: there can be neither infinitely many good blocks nor infinitely
many NWP activations.

At most $B$ corrective scans follow a complete block of $B$ ordinary
steps, so their number is at most $N$. Lemma~\ref{lem:corral} includes
their minor-loop work. Ordinary steps cost $O(Nmn)$; NWP updates
cost $O(Kmn)$. Thus the work is
$O(mn(1/(\varepsilon)+L+B))$, which is \eqref{eq:hybrid-work} for fixed
$B$ and $n\geq2$. The full product count obeys the same order.
The active set size bound and the stored current vectors give the
storage claim.
\end{proof}

A failed progress test forces NWP to advance, so the theorem does
not require WMC to become effective; omitting WMC preserves the
argument.
By Lemma~\ref{lem:reciprocal}, an STA block with
$\sum a_k\geq2B$ passes the reciprocal test. Accepted WMC points
change the trajectory, so the cumulative condition for a separate
ordinary STA run need not hold for the hybrid.

\section{STA--WMC with remembered pivots}
\label{sec:modified}

STA--WMC combines STA with persistent WMC. Unlike the hybrid, it
tries remembered STA pivots before a full corrective scan, adds a
relative-progress trigger, and omits the NWP safeguard.
This variant uses STA's pivot history to guide geometric correction.
The experiments assess the combined effect of pivot reuse and the
correction schedule.

\subsection{Algorithm}

Let $C\subseteq S$ store every distinct pivot selected by ordinary STA.
When a correction searches all of $S$ and selects a point minimizing
$a^Tv_i$, add that point to $C$ as well, even if it was not previously
chosen as an STA pivot.
Maintain a separate corral point $a=A\lambda$ using the WMC
procedure of Section~\ref{sec:hybrid}. The candidate pool $C$
can therefore be larger than the active corral.

Fix a block size $B\geq1$ and a relative-progress threshold
$\tau\in(0,1)$. The experiments use $B=8$ and $\tau=0.01$.
Initialize $C=\varnothing$, $x=\mathbf{1}/n$, and $p=Vx$.
Repeat the following cycle, stopping at an approximate answer or
a verified full-set separator.
\begin{enumerate}
\item Record $d_0=\norm p$ and perform $B$ deepest-pivot STA steps,
      collecting each new pivot in $C$. Stop early if a stopping
      test succeeds. Write $d=\norm p$ after the block.

\item Activate a WMC correction phase if
\begin{equation}\label{eq:modified-trigger}
       d^{-1}-d_0^{-1}<B
       \quad\hbox{or}\quad
       d>(1-\tau)d_0.
\end{equation}
      For $\tau=0.01$, the second condition detects less than one
      percent relative residual reduction over the block.
      If neither condition holds, resume ordinary STA.
      On the first activation, initialize the corral at the
      first remembered STA pivot.

\item During an activated correction phase, allow up to $B$ corrective
      insertion attempts. At each attempt, first offer the current
      corral point $a$ to the main iterate as described below,
      stopping if a stopping test succeeds.
      Choose $v_j\in C$ minimizing $a^Tv_j$.
      If no candidate in $C$ satisfies $a^Tv_j<\norm a^2$, scan the
      full set $S$ instead and add a minimum-score point to $C$.
      Apply the WMC insertion, affine minimization, and coefficient
      deletion procedure to an improving vertex; end the correction phase
      if none is accepted.
      Offer the resulting corral point to the main iterate and
      continue unless a stopping test or the correction-phase limit is reached.

\item Retain both $C$ and the corral, and resume ordinary STA.
\end{enumerate}

To offer $a\in\conv(S)$, compute the nearest point $q$ to the
origin on $[p,a]$. If $q=(1-t)p+ta$ has smaller norm than $p$,
replace $p$ by $q$ and update its original convex coefficients
by the same combination. This preserves feasibility and never
increases the main residual.

A separator for $\conv(C)$ alone is not a global certificate;
a reported separator must pass a full-set test. The implementation
checks the main point at its full STA scans. It stores corral
vertex indices and weights, rather than a dense matrix of their
original coefficient vectors.

\subsection{Progress and complexity}

At small residuals a block can pass the reciprocal-progress test
while making little relative progress. The additional trigger
detects this behavior but by itself supplies no unconditional accelerated guarantee.

In exact arithmetic, STA--WMC retains the ordinary quadratic
bound for approximation or separation on any input.
Each ordinary STA step satisfies the squared-residual decrease in
Proposition~\ref{prop:quadratic}, and accepted WMC refinements
never increase the residual. Thus, within $O(\varepsilon^{-2})$ ordinary
STA steps, the algorithm returns an approximate point or a separator
verified on the full input.
There are at most $B$ correction attempts per block of $B$
ordinary steps. Each candidate scan costs at most $O(mn)$,
since $C$ contains at most $n$ input points. Because the corral
is retained, Lemma~\ref{lem:corral} also bounds the total
affine minimization and factorization work over all insertion
attempts. The resulting total arithmetic bound is
$O(mn/\varepsilon^2)$.
Section~\ref{sec:computations} evaluates the combined effect of
pivot memory, WMC, and the additional trigger, rather than pivot
reuse alone.

\begin{corollary}[Conditional complexity of STA--WMC]
\label{cor:modified-conditional}
Suppose $0\in\conv(S)$. Index only the ordinary STA steps of
STA--WMC by $k=0,1,\ldots$. Let $p_k$ be the current point
immediately before step $k$, including any corrections since the
preceding ordinary step, and let $v_{i_k}$ be its chosen pivot.
Put
\[
 \delta_k=\norm{p_k},\qquad
 s_k=-p_k^Tv_{i_k}/\delta_k,\qquad
 a_k=s_k^2/\delta_k.
\]
Suppose that, for constants $\gamma>0$ and $N_0\geq1$ independent
of $\varepsilon$,
\[
 \frac1N\sum_{k=0}^{N-1}a_k\geq\gamma,\qquad N\geq N_0,
\]
as long as the algorithm has not terminated. Then
$O(N_0+1/(\gamma\varepsilon))$ ordinary STA steps suffice, and the total
arithmetic work, including corrections, is
$O(mn(N_0+1/(\gamma\varepsilon)))$.
For fixed $N_0$ and $\gamma$, these bounds are $O(1/\varepsilon)$
steps and $O(mn/\varepsilon)$ arithmetic operations.
\end{corollary}
\begin{proof}
Let $\widehat p_{k+1}$ be the point obtained by the ordinary STA
step at $p_k$. If the algorithm has not terminated, accepted
corrections give $\delta_{k+1}\leq\norm{\widehat p_{k+1}}$.
Lemma~\ref{lem:reciprocal} therefore yields
\[
 \delta_{k+1}^{-1}-\delta_k^{-1}
 \geq\norm{\widehat p_{k+1}}^{-1}-\delta_k^{-1}
 \geq\tfrac12 a_k.
\]
Summing and using the averaged condition gives
$\delta_N\leq(\delta_0^{-1}+\gamma N/2)^{-1}$ for $N\geq N_0$,
unless the algorithm has already terminated. This proves the
step bound. There are at most as many correction attempts as
ordinary steps; Lemma~\ref{lem:corral} includes all affine
minimizations and factorization updates, proving the work bound.
\end{proof}
The condition in Corollary~\ref{cor:modified-conditional} concerns
the actual STA--WMC trajectory. It is not inferred from a separate
ordinary STA run, because corrections can change subsequent
iterates and pivot choices.

\subsection{Comparison of complexity guarantees}
\label{sec:rates}
Table~\ref{tab:guarantees} summarizes the worst-case arithmetic bounds
proved here. Under the averaged progress condition along each
method's own trajectory, STA and STA--WMC additionally require
$O(N_0+1/(\gamma\varepsilon))$ ordinary steps and
$O(mn(N_0+1/(\gamma\varepsilon)))$ arithmetic work
(Theorem~\ref{thm:amortized} and
Corollary~\ref{cor:modified-conditional}). In particular, fixed
$N_0$ and $\gamma$ give conditional $O(1/\varepsilon)$ and $O(mn/\varepsilon)$
bounds, respectively. The Nesterov certificate distinction is
discussed in Section~\ref{sec:comparison}.

\begin{table}[htbp]
\centering\small
\caption{Arithmetic work to obtain an approximate point or a strict
separator. The bounds for STA, STA--WMC, NWP, and the hybrid apply to
all inputs. The stated Nesterov bound assumes $0\in\conv(S)$.
The hybrid and STA--WMC use fixed $B$.}
\label{tab:guarantees}
\begin{tabular}{lll}
\toprule
Method & Scope of bound & Arithmetic work\\
\midrule
STA & All inputs & $O(mn/\varepsilon^2)$\\
STA--WMC & All inputs & $O(mn/\varepsilon^2)$\\
Nesterov & $0\in\conv(S)$ & $O((mn+n\log n)\sqrt{\log n}/\varepsilon)$\\
NWP & All inputs & $O(mn\sqrt{\log n}/\varepsilon)$\\
Hybrid & All inputs & $O(mn\sqrt{\log n}/\varepsilon)$\\
\bottomrule
\end{tabular}
\end{table}

The hybrid preserves the NWP guarantee while permitting STA and WMC
progress. The experiments assess the computational value of that
combination relative to ordinary STA, NWP, Nesterov, and STA--WMC.

\section{Computational comparisons}
\label{sec:computations}
We compare STA, the hybrid, entropic Nesterov acceleration, NWP,
and STA--WMC on convex-hull queries derived from 14 public
classification datasets \cite{LIBSVMData}. For each dataset, we test
whether the convex hulls of its two labeled groups intersect or
admit a strict separator. The spherical input consists of normalized
signed augmented columns $(a_i,1)$ and $(-b_j,-1)$.
We use the full supplied training files and feature representations;
these are geometric feasibility tests, not prediction-accuracy tests.

Each method was run three times on each public-data query, with
tolerance $\varepsilon=10^{-6}$ and a nominal five-second allowance per run,
giving 210 runs in total. Methods were run one at a time, with their
order varied between repetitions. Reported times include initialization,
updates, and stopping tests, but exclude data preparation and independent
verification of the answers. All methods used the same computing
environment; the times reflect these implementations and this environment.
The hybrid and STA--WMC use $B=8$, and STA--WMC uses $\tau=0.01$.
The three benchmark batches were selected before their respective timed
outcomes; all their public classification datasets are reported here.

\begin{table}[htbp]
\centering\small\setlength{\tabcolsep}{3pt}
\caption{Public-data benchmarks: median seconds over three runs.
$m$ includes the augmented coordinate; $n$ is the number of observations.
Superscript S denotes separation; other successful outputs are approximate
convex combinations. A dash means no successful run; $*$ marks the sole successful run
(4.817 seconds), with the other two runs timing out at five seconds.
The final row counts datasets completed in all three runs.}
\label{tab:benchmarks}
\begin{tabular}{llrrrrr}
\toprule
Dataset & $m\times n$ & STA & Hybrid & Nesterov & NWP & STA--WMC\\
\midrule
Australian credit & $15\times690$ & 1.09 & 0.00665 & 0.335 & 2.16 & 0.00638\\
Adult (a9a) & $124\times32561$ & -- & 2.27 & -- & -- & 2.97\\
Heart & $14\times270$ & 0.0132 & 0.0123 & 0.0619 & 0.237 & 0.0176\\
Sonar$^{\mathrm S}$ & $61\times208$ & 0.164 & 0.0288 & 0.344 & 0.112 & 0.0288\\
Ionosphere & $35\times351$ & -- & 0.0608 & 0.554 & -- & 0.0119\\
German credit & $25\times1000$ & 0.0197 & 0.0204 & 0.145 & 0.602 & 0.0258\\
Diabetes & $9\times768$ & 0.00182 & 0.00193 & 0.038 & 0.253 & 0.0027\\
w8a & $301\times49749$ & -- & 0.112 & -- & -- & 0.0984\\
IJCNN1 & $23\times35000$ & 0.142 & 0.128 & 2.33 & -- & 0.134\\
cod-rna & $9\times59535$ & -- & 0.0333 & -- & -- & 0.0187\\
svmguide1 & $5\times3089$ & 3.55 & 0.0028 & 4.817$^{*}$ & 2.61 & 0.00213\\
svmguide3 & $23\times1243$ & 2.3 & 0.0427 & 1.49 & -- & 0.0231\\
Splice & $61\times1000$ & 3.28 & 1.61 & 0.773 & 2.06 & 0.0417\\
Colon cancer$^{\mathrm S}$ & $2001\times62$ & 0.000781 & 0.00556 & 0.00506 & 0.00512 & 0.00824\\
\midrule
Completed & & 10/14 & 14/14 & 10/14 & 8/14 & 14/14\\
\bottomrule
\end{tabular}
\end{table}

The hybrid and STA--WMC completed every public-data case in
all repetitions. STA and Nesterov completed ten consistently, and
NWP eight; Nesterov also completed \texttt{svmguide1} in one repetition
near the deadline. Every returned certificate passed independent
reconstruction or full-input separation checks. Adult, w8a, and
cod-rna contain identical feature
vectors with opposite labels, so their size need not imply difficulty.

STA is fastest among the five methods on the diabetes and
colon-cancer cases. The hybrid and STA--WMC substantially improve
completion on several cases where ordinary STA does not finish.
Their relative times vary: on Adult the medians are 2.27 seconds
for the hybrid and 2.97 for STA--WMC, whereas on Splice they are
1.61 and 0.0417 seconds. Nesterov takes 0.773 seconds on Splice.
Thus the hybrid's unconditional guarantee does not imply the
smallest measured runtime. On Sonar both the hybrid and STA--WMC
take approximately 0.029 seconds; small or overlapping timing
differences do not establish superiority.

\paragraph{Controlled geometric comparison.}
To examine geometry directly, Table~\ref{tab:overall} summarizes
204 synthetic inputs with known membership or separation certificates:
66 random membership, 72 boundary membership, and 66 separation cases.
The three classes are defined as follows.
\begin{itemize}
\item \emph{Random membership.} Generate vectors with independent
standard normal (Gaussian) coordinates and normalize them to unit length.
Assign positive lognormal weights, then append the unit vector opposite
their weighted sum. This construction guarantees that the origin is
in the convex hull.
\item \emph{Boundary membership.} Place unit points whose convex hull contains
the origin in the hyperplane with last coordinate zero, and add unit points with
positive last coordinate. The origin is therefore on the boundary of
the full hull. The points in the hyperplane are stretched in different
coordinate directions and renormalized; additional cases use a stretched simplex or a
collection contained in a lower-dimensional subspace.
\item \emph{Separation.} Starting from either membership construction
in one lower dimension, replace each unit point $u_i$ by
$(\sqrt{1-h^2}\,u_i,h)$. The resulting hull has distance exactly
$h>0$ from the origin. We use $h\in\{0.05,10^{-3},10^{-5}\}$
and tolerance $\varepsilon=h/10$, so these cases require a separator.
\end{itemize}
Here $m\leq1000$, $n\leq20001$, and tolerances reach $10^{-7}$.
Membership and separation allowances are 20 and 10 seconds,
respectively. These completion counts are distinct from the
five-second public-data study.

\begin{table}[htbp]
\centering
\caption{Completed synthetic cases within the study allowances.
Random membership is constructed to contain the origin; boundary
membership places it on a supporting face; separation places it outside
the hull at a prescribed positive distance.}
\label{tab:overall}
\small\setlength{\tabcolsep}{5pt}
\begin{tabular}{lrrrr}
\toprule
Method & \shortstack{Random\\membership} & \shortstack{Boundary\\membership} & Separation & Total\\
\midrule
STA & 66/66 & 0/72 & 44/66 & 110/204\\
Hybrid & 66/66 & 60/72 & 66/66 & 192/204\\
Nesterov & 66/66 & 26/72 & 55/66 & 147/204\\
NWP & 55/66 & 0/72 & 44/66 & 99/204\\
STA--WMC & 66/66 & 72/72 & 66/66 & 204/204\\
\bottomrule
\end{tabular}
\end{table}

Ordinary STA completes all random membership cases but none of the
boundary cases at the requested tolerances. STA--WMC completes
all 204; the hybrid completes 192, including every separation case.
In the five-method comparison, Nesterov is faster than the STA-based
methods on the highest-dimensional random configuration. These
controlled results support a geometric interpretation of performance;
the public-data tests alone do not identify distance from the boundary.

The synthetic experiments include separate timing sessions. We therefore
use their completion counts and broad timing patterns, rather than infer
rankings from small timing differences.

\section{Conclusions and practical guidance}
\label{sec:practice}
The cumulative opposite-depth analysis gives a conditional
$O(1/\varepsilon)$ regime for ordinary STA while allowing shallow pivots.
The same conditional rate holds for STA--WMC when the averaged
condition is satisfied along its actual sequence of ordinary steps;
the total work is $O(mn/\varepsilon)$ for fixed condition constants.
For ordinary STA, the three-point example establishes the sharp unconditional
$\Theta(\varepsilon^{-2})$ obstruction. Identifying boundary input classes
with uniform cumulative-depth constants remains an open direction.
The hybrid combines STA, WMC, and NWP with an unconditional
$O(mn\sqrt{\log n}/\varepsilon)$ work bound, including affine corrections.

The experiments show complementary strengths. STA is effective on
favorable interior and separation instances, while STA--WMC completes
all 204 controlled synthetic cases and all public-data cases.
Nesterov is faster than the STA-based methods on the highest-dimensional
random configuration. The hybrid completes all public-data cases and
retains its unconditional accelerated bound, but its runtime varies
and it does not complete every synthetic case within
budget. Practical comparisons should therefore report both absolute
time and completion, alongside the theoretical guarantee.

\appendix
\section{Supporting proofs for Nesterov acceleration}
\label{sec:synthetic-details}
For completeness, we give a self-contained proof of the Nesterov
convergence estimate for the spherical convex hull problem and explain
how its simplex update is computed. Throughout this appendix,
\[
 f(x)=\tfrac12\norm{Vx}^2,\qquad
 d(x)=\log n+\sum_{i=1}^n x_i\log x_i,\qquad x\in\Delta_n,
\]
with $0\log0=0$. We use the $\ell_1$ norm on the simplex.
The vectors $x_k,y_k,z_k$ are the Nesterov iterates
\eqref{eq:nesterov-hard}--\eqref{eq:nesterov-coupling}, distinct from
the Nemirovski iterates.

\subsection{Nesterov acceleration estimate}
\label{app:nesterov-estimate}
Since each column of $V$ has unit Euclidean norm,
$\norm{Vh}\leq\sum_i|h_i|=\norm h_1$. Expanding the quadratic gives
\begin{equation}\label{eq:nesterov-smooth-model}
 \begin{aligned}
 f(y)&=f(x)+\nabla f(x)^T(y-x)+\tfrac12\norm{V(y-x)}^2\\
     &\leq f(x)+\nabla f(x)^T(y-x)+\tfrac12\norm{y-x}_1^2.
 \end{aligned}
\end{equation}
Thus $1$ is a valid smoothness constant.
As shown in Section~\ref{sec:comparison}, entropy $d$ has minimum zero
at $x_0=\mathbf{1}/n$, is bounded above by $\log n$, and is
$1$-strongly convex in $\ell_1$.
Put $\omega_k=(k+1)/2$ and $A_k=\sum_{i=0}^k\omega_i=(k+1)(k+2)/4$.

\begin{proposition}\label{prop:nesterov-estimate}
For the Nesterov iterates in Section~\ref{sec:comparison} and any minimizer
$x^*$ of $f$ on $\Delta_n$,
\[
 f(y_k)-f(x^*)\leq\frac{d(x^*)}{A_k}
 \leq\frac{4\log n}{(k+1)(k+2)}.
\]
On membership inputs this gives the residual bound
\eqref{eq:nesterov-bound}.
\end{proposition}
\begin{proof}
Let $g_i=\nabla f(x_i)$ and define
$\ell_i(x)=f(x_i)+g_i^T(x-x_i)$ and
$\Psi_k(x)=d(x)+\sum_{i=0}^k\omega_i\ell_i(x)$.
Since $G_k=\sum_{i=0}^k\omega_i g_i$, minimizing $\Psi_k$
gives the exponential formula for $z_k$ in \eqref{eq:nesterov-soft}.
We prove by induction that
\begin{equation}\label{eq:nesterov-invariant}
 \Psi_k(z_k)\geq A_k f(y_k).
\end{equation}
Strong convexity gives $d(x)\geq\tfrac12\norm{x-x_0}_1^2$.
Using $A_0=\omega_0=1/2$,
\[
 \Psi_0(x)\geq\tfrac12\bigl[
 f(x_0)+g_0^T(x-x_0)+\tfrac12\norm{x-x_0}_1^2\bigr]
 \geq\tfrac12 f(y_0).
\]
The last inequality follows by minimizing the bracket over $\Delta_n$
and applying \eqref{eq:nesterov-smooth-model} at $y_0$.

Assume \eqref{eq:nesterov-invariant} at $k$.
Strong convexity of $\Psi_k$, convexity of $f$, and the coupling
\eqref{eq:nesterov-coupling} imply
\[
 \Psi_{k+1}(x)\geq A_{k+1}f(x_{k+1})
 +\omega_{k+1}g_{k+1}^T(x-z_k)+\tfrac12\norm{x-z_k}_1^2.
\]
For $x\in\Delta_n$, set
$v=(A_k y_k+\omega_{k+1}x)/A_{k+1}\in\Delta_n$.
Then $v-x_{k+1}=\omega_{k+1}(x-z_k)/A_{k+1}$ and
$\omega_{k+1}^2\leq A_{k+1}$, so
\[
 \begin{aligned}
 \Psi_{k+1}(x)
 &\geq A_{k+1}\bigl[f(x_{k+1})+g_{k+1}^T(v-x_{k+1})
                  +\tfrac12\norm{v-x_{k+1}}_1^2\bigr]\\
 &\geq A_{k+1}f(y_{k+1}).
 \end{aligned}
\]
The last step uses the minimizing property of $y_{k+1}$ and
\eqref{eq:nesterov-smooth-model}. This proves the induction.
Finally, convexity gives $\ell_i(x^*)\leq f(x^*)$, hence
$A_k f(y_k)\leq\Psi_k(z_k)\leq\Psi_k(x^*)
\leq d(x^*)+A_k f(x^*)$.
Divide by $A_k$; when $0\in\conv(S)$, use $f(x^*)=0$.
\end{proof}

\subsection{Computing the simplex update}
\label{app:hard}
The Nesterov update for $y_k$ in Section~\ref{sec:comparison}
minimizes a linear term plus half the squared $\ell_1$ distance from
$x_k$. It can be computed by transferring weight from coordinates
with larger gradient entries to one with the smallest entry.
For this calculation, write $x=x_k$ and $g=g_k$.
\begin{enumerate}
\item Choose $j\in\argmin_i g_i$ and put $a=g_j$.
Order the indices $i_1,\ldots,i_r$ with $x_{i_\ell}>0$ and
$g_{i_\ell}>a$ by decreasing $g_{i_\ell}$.
If none exist, return $y_k=x$.
\item Put $C_\ell=\sum_{h\leq\ell}x_{i_h}$ and
$P_\ell=\sum_{h\leq\ell}x_{i_h}g_{i_h}$, with $C_0=P_0=0$.
Transferring total weight $t\in[C_{\ell-1},C_\ell]$ from these
coordinates in order to coordinate $j$ gives objective value
\[
 \phi_\ell(t)=2t^2+at-P_{\ell-1}-(t-C_{\ell-1})g_{i_\ell}.
\]
Its minimum on that interval occurs at
$t_\ell=\min\{C_\ell,\max\{C_{\ell-1},(g_{i_\ell}-a)/4\}\}$.
\item Choose the candidate with smallest $\phi_\ell(t_\ell)$.
Perform that transfer to obtain $y_k$.
\end{enumerate}
To verify optimality, for any feasible change $y-x$, put
$t=\norm{y-x}_1/2$. This is the total weight removed and also the
total weight added, so the quadratic cost is $2t^2$.
For fixed $t\leq C_r$, the smallest linear cost is attained by
removing weight from the largest available entries of $g$ and adding
it at $j$. This is feasible since $C_r\leq1-x_j$ and gives precisely
$\phi_\ell(t)$ on the indicated interval.
A transfer exceeding $C_r$ cannot reduce the linear term below
$aC_r-P_r$ and has larger quadratic cost than the transfer $C_r$,
so it cannot be optimal. If there are no indices in Step~1, every
feasible linear change is nonnegative and $y_k=x$ is optimal.
Sorting costs $O(n\log n)$; all remaining work and auxiliary
storage are $O(n)$.

\section*{Acknowledgments}
The author acknowledges the assistance of ChatGPT (OpenAI) in
mathematical discussions, the development of proofs, the implementation
of computational experiments, and the preparation and revision of the
manuscript. The author assumes responsibility for all content of
the article.

\end{document}